\documentclass{article}
\usepackage{graphicx} 
\usepackage{amsmath, amssymb, amsthm}
\usepackage{hyperref}
\usepackage{xcolor}
\usepackage{comment}
\usepackage{tikz-cd}
\usepackage[left=3.4cm, right=3.4cm]{geometry}
\usepackage{xurl}

\usepackage{algorithm}
\usepackage{algpseudocode}

\newcommand{\F}{\mathbb{F}}
\newcommand{\Q}{\mathbb{Q}}

\newcommand{\calE}{\mathcal{E}}
\newcommand{\calR}{\mathcal{R}}

\DeclareMathOperator{\GRS}{GRS}
\DeclareMathOperator{\diag}{diag}
\DeclareMathOperator{\shortening}{Sh}
\DeclareMathOperator{\puncturing}{Pct}

\newtheorem{theorem}{Theorem}

\newtheorem{lemma}[theorem]{Lemma}

\newtheorem{corollary}[theorem]{Corollary}

\theoremstyle{definition}

\newtheorem{definition}[theorem]{Definition}
\theoremstyle{remark}
\newtheorem{remark}[theorem]{Remark}
\newtheorem{example}[theorem]{Example}

\title{An Attack on High Rate McEliece Cryptosystems Using Generalized Reed Solomon Codes with Weight $2$ Mask}
\author{Julia Lieb, Abhinaba Mazumder, Michael Schaller}
\date{\today}

\begin{document}

\maketitle

\begin{abstract}
    Due to the insecurity of McEliece cryptosystems instantiated with Generalized Reed-Solomon codes, there have been several proposals of McEliece type systems that replace the permutation matrix by a matrix $M$ with larger row and column weight.
    In many of them, the secret key is still a GRS code.
    There have been successful attacks on some of those schemes with row and column weight between $1$ and $1 + R$, where $R$ is the rate of the code. The case of weight two and larger has been left open in these works.
    Subsequently, several authors proposed schemes with weight exactly two and with even higher weight.
    We provide distinguishers for the public codes appearing in these cryptosystems in the high rate regime.
    In addition, we give a framework to turn a good enough distinguisher into a key-recovery attack.
    In the case where the matrix $M$ has row and column weight $2$, we can successfully attack the scheme in the high rate regime using a cube code distinguisher.
\end{abstract}

\section{Introduction}

After Sidelnikov and Shestakov \cite{sidelnikov1992insecurity} broke the McEliece cryptosystem instantiated with generalized Reed-Solomon (GRS) codes \cite{niederreiter1986knapsack}, there have been proposals to modify the cryptosystem in order to still use GRS codes for the private key \cite{almeida2024convolutional, baldi2011variant, baldi2016enhanced, weger16}.
The modifications are mainly based on using as a mask a matrix with higher row and column weight instead of a permutation matrix in the McEliece cryptosystem with GRS codes.
The schemes \cite{baldi2011variant, baldi2016enhanced} which use row and column weight between $1$ and $1 + R$ where $R$ is the rate of the code have been successfully attacked in \cite{couvreur2014distinguisher, couvreur2015polynomial}.
The case of row and column weight equal to $2$, which is used in \cite{weger16, bolkema2017variations, almeida2024convolutional}, has been left open in \cite{couvreur2014distinguisher, couvreur2015polynomial}.
A similar scheme with higher weight was proposed in \cite{Dyseryn26}.
We provide a framework that shows that having access to a certain distinguisher, we are able to attack the McEliece cryptosystem with GRS codes also when a mask with row and column weight two is used as proposed in \cite{weger16}.
Under certain heuristic assumptions, we obtain that the cube code can serve as such a distinguisher for high enough code rates.
Using shortening, one can improve the range of the distinguisher to lower rates.
However, extending it to all rates remains an open problem.
In the high rate regime without shortening where $n > 2 k^2 -4 k + 4$, where $k$ is the dimension of the dual code, we are able to turn the distinguisher into a key-recovery attack.

The distinguisher only uses column weight equal to $2$, while key recovery becomes a bit more complicated if only the column weight is $2$ without the row weight also being $2$.
For this reason, we only claim an attack on \cite{weger16} in the high rate regime.
However, it seems likely that it is possible to adapt the attack to the case of only column weight equal to $2$ and one can certainly recover some information about the private key of \cite{almeida2024convolutional, bolkema2017variations} in the high rate regime.

The basic idea for the attack in the high rate regime is similar to the ideas in \cite{couvreur2014distinguisher, couvreur2015polynomial} in the sense that we use a distinguisher in order to recover information about the private key.
We are given a public code of the form 
\[
    \GRS_k(P, \mu) M
\]
where the evaluation points and column multipliers of the GRS code are secret and $M$ is a secret invertible matrix with row and column weight $2$.
In order to recover the mask $M$, we multiply a generator matrix of the public code $\GRS_k(P, \mu)  M$ on the right with a matrix $N$ that corresponds to an elementary column operation adding the multiple of one column to another column.
In this way we get the code 
\[
    \GRS_k(P, \mu) M N.
\]
We apply a distinguisher to the code  $\GRS_k(P, \mu) M N$ to find among the matrices $N$ those for which the number of non-zero entries in $M N$ does not exceed the number of non-zero entries in $M$.
Since all columns of $M$ are of weight $2$, this corresponds to finding cancellation patterns when performing column operations on $M$.
This information about the mask $M$ allows us to recover all of $M$ (or rather a matrix which is in some sense equivalent) so that we can apply the Sidelnikov-Shestakov \cite{sidelnikov1992insecurity} attack or a filtration attack \cite{couvreur2014distinguisher} to recover a GRS code with which we can decode the encrypted message.

The paper is organized as follows.
In Section \ref{sec:basics} we recall the necessary background from coding theory.
In Section \ref{sec:Desc_McEliece_GRS_Weight_Two} we describe the McEliece cryptosystem with GRS codes and weight $2$ mask.
In Section \ref{sec:cube_code_distinguisher} we show that for high enough rate the cube code provides a distinguisher for the public codes in the cryptosystem from random codes.
This gives evidence that we can use the cube code as the distinguisher for cancellation as described in the introduction.
We also show an improvement on the distinguisher using shortening.
Section \ref{sec:framework_distinguishing_key_recovery}
is at the heart of the paper.
First we describe the framework showing how a distinguisher as described in the introduction leads to a key-recovery attack.
Then we show how we can instantiate it in the high rate regime with the cube code.
We conclude this section with the experiments of our attack.
In Section \ref{sec:outlook} we describe several interesting open problems.

\section{Basics of Linear Codes}
\label{sec:basics}

\begin{definition}
    Let $v = (v_1, \ldots, v_n), w = (w_1, \ldots, w_n) \in \F_q^n$.
    We define the \textbf{standard inner product} between $v$ and $w$ as
    \[
        \langle v, w \rangle := \sum_{i=1}^n v_i w_i.
    \]
\end{definition}

\begin{definition}
Let $k \leq n$ be positive integers and let $C$ be an $[n,k]_q$ linear code over $\F_q$, i.e., $C$ is a $k$-dimensional subspace of $\mathbb F_q^n$. Then, the linear $[n,n-k]_q$ code
$$C^\perp := \{ x \in \F_q^n \mid \langle c, x \rangle = 0 \ \forall \  c\in C \}$$
is called the \textbf{dual code} of $C$.
\end{definition}

\begin{definition}\label{def:reedsolomon}
Let  $k \leq n \leq q$ be positive integers. Let $P=(P_1, \ldots, P_n) \in \F_q^n$ with $P_i \neq P_j,$ for all $i \neq j \in \{1, \ldots, n\}$. Let $\mu=(\mu_1, \ldots, \mu_n)\in \F_q^n$ with $ \mu_i \neq 0$ for all $i \in \{1, \ldots, n\}.$  Then,
\begin{equation*}
\GRS_{k}(P,\mu)= \left\{ (\mu_1 f(P_1), \ldots, \mu_n f(P_n)) \bigm| f \in \F_q[x], \ \text{deg}(f) <  k \right\}
\end{equation*}
is called a \textbf{Generalized Reed-Solomon (GRS)} code of length $n$ and dimension $k$ over $\mathbb F_q$.
\end{definition}

\begin{lemma}[\cite{pellikaan2017codes}, Prop. 5.1.26]
\label{lemma:dual_grs}
Let $k \leq n\leq q$ be positive integers. Then,\begin{equation*}
\GRS_{k}(P,\mu)^{\perp} = \GRS_{n-k}(P,\tilde{\mu}),
\end{equation*}
where 
\begin{equation*}
\tilde{\mu}_i =  \mu_i^{-1}  \prod_{\substack{ j= 1 \\ j \neq i}}^n(P_i-P_j)^{-1}.
\end{equation*}    
\end{lemma}

In the following, we introduce power codes. The special case of square codes has become prominent as a distinguisher in GRS code based McEliece type cryptosystems; see e.g. \cite{couvreur2014distinguisher, couvreur2015polynomial}.

\begin{definition}
Let $v=(v_1, \ldots, v_n), w=(w_1, \ldots, w_n) \in \mathbb{F}_q^n$. The \textbf{Schur product} of $v$ and $w$ is defined as
$$ v\ast w := ( v_1 w_1, \ldots, v_n w_n ).$$
Let $C_1,C_2 \subset \mathbb{F}_q^n$ be two linear codes.
The \textbf{Schur product} of $C_1$ and $C_2$ is defined as 
$$ C_1 \ast C_2 := \langle \{ c_1 \ast c_2 \ | \ c_1 \in C_1,\ c_2 \in C_2 \} \rangle \subset \mathbb{F}_q^n.$$
For a linear code $C \subset \mathbb{F}_q^n$, we call $C * C$ the \textbf{square code} of $C$ and denote it by $C^{(2)}$. 
More generally, we call $C^{(w)}:=\underbrace{C\ast\cdots\ast C}_{w\ \text{times}}$ the \textbf{$w$-th order power code of $C$}. In particular, we call $C^{(3)}$ the \textbf{cube code} of $C$.
\end{definition}

The reason why the square code can serve as a distinguisher for GRS code based McEliece cryptosystems where the mask has row and column weight at most $1+R$ where $R$ is the code rate, is that the dimension of the square code of a GRS code is smaller than the dimension of the square code of a random code.
It is easy to see that 
$$\dim(GRS_k(P,\mu)^{(2)})= \min\left\{2k-1,n\right\}. $$
The following theorem provides the dimension of the square code of a random code.

\begin{theorem}[\text{\cite[Theorem 2.3]{cascudo2015squares}}]
For a random linear code $C \subset \mathbb{F}_q^n$ of dimension $k$, one has with high probability that
$$\dim(C^{(2)})= \min\left\{\binom{k+1}{2},n\right\}. $$
\end{theorem}

In this paper, we will employ the cube code as a distinguisher for GRS codes with a weight two mask.
It is conjectured that the dimension of the cube code of a random code of dimension $k$ and length $n$ is with high probability equal to $$\min\left\{\binom{k+2}{3},n\right\}.$$
To the best of our knowledge, a proof of this result is still missing.

In Section \ref{sec:cube_code_distinguisher}, we will show that the dimension of the cube code of a weight two masked GRS code of dimension $k$ and length $n$ is upper bounded by $$\min\{n, 2 k^2 - 4k + 4\}.$$

Later in the paper we will also need the following lemma about power codes of sums of codes, which is easy to see.

\begin{lemma}\label{lem:power_codes}
  Let $m,w,n\in\mathbb N$ and let $C_1,\hdots,C_m\subset\mathbb F_q^n$ be linear codes. Then,
$$    \left(\sum_{i=1}^m C_i\right)^{(w)}\subseteq\sum_{i_1,\hdots,i_w\in\{1,\hdots,m\}}C_{i_1}\ast\cdots\ast C_{i_w}$$
\end{lemma}

\begin{definition}
    Let $C$ be an $[n, k]_q$ linear code.
    Let $J \subseteq \{1, \ldots, n\}$.
    We define the \textbf{punctured code} of $C$ in the coordinates of $J$ to be
    \[
        \puncturing_J(C) := \{(c_i)_{i \notin J} \bigm| c \in C\}.
    \]
\end{definition}

\begin{definition}
    Let $C$ be an $[n,k]_q$ linear code and let $J \subseteq \{1, \ldots, n\}$.
    We define the code
    \begin{equation}
    \label{eq:code_of_zeros_J}
        C(J) := \{c \in C: c_i = 0 \text{ for all } i \in J\}.
    \end{equation}
    We define the \textbf{shortened code} of $C$ in the coordinates of $J$ to be
    \[
        \shortening_J(C) := \puncturing_J (C(J)).
    \]
\end{definition}

\section{Description of the McEliece Cryptosystem with GRS Codes and Weight Two Mask}
\label{sec:Desc_McEliece_GRS_Weight_Two}

In the following, we describe the setup of the weight two masked McEliece cryptosystem with GRS codes \cite{weger16}.
Note that we describe the codes here in a slightly different manner using duals and transposes to make the exposition simpler.
Both representations can be derived and computed in a simple manner from each other using Lemma \ref{lemma:dual_grs}.

The private key consists of the following data:

\begin{itemize}
    \item A secret GRS code $C_{sec} = \GRS_k(P, \mu)^\perp$ over $\F_q$ of length $n$, dimension $n - k$, the dual of a GRS code with vector of evaluation points $P$ and vector of column multipliers $\mu$ with an efficient decoding algorithm and generator matrix $G_{sec}\in\mathbb F_q^{(n - k)\times n}$.
    \item An $(n - k) \times (n - k)$ change of basis matrix $S$ over $\mathbb F_q$.
    \item An invertible $n \times n$ matrix $M$ over $\mathbb F_q$ of row and column weight two, i.e., a matrix in which each row and column has exactly two non-zero entries.
\end{itemize}
The public key then consists of the matrix $G_{pub} = S G_{sec} (M^T)^{-1}$.

Next, we explain encryption and decryption.
Let $t_{GRS} = \left\lfloor \frac{k}{2} \right\rfloor$ be the error correction capability of the GRS code $C_{sec}$.
We set
\[
    t = \left\lfloor \frac{t_{GRS}}{2} \right\rfloor.
\]
Now, in order to encrypt, the sender who wants to send message $m \in \F_q^{n-k}$ sends
\[
    c = m G_{pub} + e
\]
where $e$ is a random vector of weight $t$.

The receiver of the ciphertext $c$ decrypts as follows.
She multiplies $c$ by $M^{T}$ to get $c M^{T} = m S G_{sec} + eM^{T}$ with the weight of $e M^{T}$ being at most $t_{GRS}$.
Therefore, she can decode using the GRS decoder for $C_{sec}$ to get $mS$ and multiplying that by $S^{-1}$ yields the original message $m$.

To switch between a weight two masked code and its dual, we use the following well-known lemma.
\begin{lemma}
    Let $C$ be an $[n, n - k]_q$-code.
    Let $D = C (M^T)^{-1}$ for an $n \times n$ invertible matrix $M$.
    Then we have
    \[
        D^\perp = C^\perp M.
    \]  
\end{lemma}
\begin{proof}
    Let $v^\perp = c^\perp M$ for $c^\perp \in C^\perp$ and let $v = c (M^T)^{-1}$ for $c \in C$.
    Then we have
    \[
        v^\perp v^T = c^\perp M ((M^T)^{-1})^T c^T = c^\perp c^T = 0.
    \]
    Hence $v^\perp \in D^\perp$ which yields the inclusion $C^\perp M \subseteq D^\perp$.
    Since the dimensions of $D^\perp$ and $C^\perp$ are the same, equality follows.
\end{proof}

Applying the previous lemma to $C = \GRS_k(P, \mu)^\perp$ and $D$ the code generated by $G_{pub}$ one obtains
\[
    D^\perp = \GRS_k(P, \mu) M.
\]
Here, $M$ is a matrix where all rows and columns have weight exactly two.
From now on we will only work with this code $\GRS_k(P, \mu) M$.
If we can recover a GRS code and a weight two mask from this code, we can also recover the information for the duals.

In the following, for $v=(v_1,\hdots,v_n)\in\mathbb F_q^n$, we define 
$$\diag(v)=\begin{pmatrix}
    v_1 &  & 0\\
    & \ddots & \\
    0 & & v_n
\end{pmatrix}.$$
With this, we can write 
\begin{align}\label{P}
M=\Pi_1 (\diag(\alpha)+\Pi \cdot \diag(\beta))
\end{align}
with $n\times n$ permutation matrices $\Pi_1$, $\Pi$ and $\alpha=(\alpha_1,\hdots, \alpha_n), \beta=(\beta_1,\hdots, \beta_n)\in (\mathbb F_q^\ast)^n$,
where the permutation corresponding to the permutation matrix $\Pi$ has no fixed points.
By abuse of notation, we write $\Pi$ for both the permutation matrix and the corresponding permutation.
Due to the equation
\begin{equation}
\label{eq:monomial_change}
    \GRS_k(P, \mu) M = \GRS_k(P, \mu) A^{-1} A M
\end{equation}
which holds for every monomial matrix $A$, we can work with the equivalent GRS code $\hat{C}=\GRS_k(P, \mu) A^{-1}$ and the corresponding weight two mask $\hat{M}=A M$ instead.
In the following, we will use this several times for different monomial matrices $A$.
By abuse of notation, we will call the resulting codes and matrices again $\hat{C}$ and $\hat{M}$ respectively.
It will be clear from the context what $\hat{C}$ and $\hat{M}$ will be.

First, we use \eqref{eq:monomial_change} with $A=\Pi_1^{-1}$ to obtain
\begin{equation}
\label{eq:mask_adding_i_and_pi_i}
    \hat{M}=\diag(\alpha)+\Pi \cdot \diag(\beta).
\end{equation}

Let $f\in\mathbb F_q[x]_{<k}$ represent a codeword $(\mu_1f(P_1),\hdots,\mu_nf(P_n))$ of $\GRS_k(P, \mu)$.
By abuse of notation we set 
$$(\mu_1f(P_1),\hdots,\mu_nf(P_n)):=(\mu_1f(P_1),\hdots,\mu_nf(P_n))\Pi_1\in\GRS_k(P, \mu)\Pi_1=\hat{C}.$$

We apply \eqref{eq:monomial_change} a second time, now with $A=\diag (\mu)$ to obtain that 
$$\hat{C}=\GRS_k(P, \mu)\Pi_1 \diag(\mu)^{-1}$$
is a Reed-Solomon code with evaluation points $P_1,\hdots, P_n$ reordered according to $\Pi_1$.
By abuse of notation we call them $P_1, \ldots, P_n$ again.
Moreover, we get
$$\hat{M}=\diag(\mu\ast \alpha)+\diag(\mu)\cdot \Pi\cdot \diag(\beta) .$$
In order to be able to work with a representation for the codewords of the public code that is as easy as possible, we consider in the following the equivalent code
$$\GRS_k(P, \mu)M \diag(\mu\ast\alpha)^{-1}=\hat{C}\hat{M} \diag(\mu\ast\alpha)^{-1},$$
whose codewords are of the form
\begin{equation}\label{eq:codewords_equivalent_code}
    (f(P_1)+\nu_1f(P_{\Pi(1)}),\ldots, f(P_n)+\nu_n f(P_{\Pi(n)}))
\end{equation}
for some $\nu\in (\mathbb F_q^\ast)^n$.
In this way, we can describe codewords in $\GRS_k(P, \mu) M \diag(\mu\ast\alpha)^{-1}$ as evaluations of polynomials $f(x_1) + yf(x_2) \in \F_q[x_1, x_2, y]$.
Note that for any invertible diagonal matrix $D\in\mathbb F_q^{n\times n}$, one has 
$$(C\cdot D)^{(3)}=C^{(3)}\cdot D^3,$$
i.e. the cube codes of $\GRS_k(P, \mu) M \diag(\mu\ast\alpha)^{-1}$ and 
$\GRS_k(P, \mu) M $ have the same dimension.

\begin{remark}
    We use the representation in terms of an equivalent code only to simplify theoretical arguments.
    For the actual attack we will never have to compute this equivalent code.
\end{remark}

This will be used in the next section to show that the cube code of such codes has smaller dimension than the cube code of a random code with the same parameters if the dimension $k$ is sufficiently small.

\section{Cube Code Distinguisher}
\label{sec:cube_code_distinguisher}

In order to employ the cube code as a distinguisher for the public code in the McEliece cryptosystem with GRS codes and weight two mask, we use the model of the public code as a certain evaluation code as described at the end of the previous section and we show that such evaluation codes can be distinguished using cube codes or higher-order power codes.

Recall that we consider a public code of the form
\[
    \GRS_k(P, \mu) M
\]
where $M$ is a matrix with row and column weight $2$.
We define the following space of polynomials in three variables:
\begin{align}
\label{eq:def_S_k}
    S_k = \{f(x_1) + y f(x_2) \in \F_q[x_1, x_2,  y]: f \in \F_q[x]_{<k} \}.
\end{align}
Note that we take the same polynomial $f$ in both terms.
$S_k$ is a $k$-dimensional vector space over $\mathbb F_q$ with basis 
\begin{align}\label{basis}
    \{1 + y, x_1 + y x_2, \ldots,  x_1^{k-1} + y x_2^{k-1}\}.
\end{align}
Evaluating the polynomials in $S_k$ at the $n$ points
\[
    (P_1, P_{\Pi(1)}, \nu_1), \ldots, (P_n, P_{\Pi(n)}, \nu_n)\in\mathbb F_q^3
\]
gives exactly the code $\GRS_k(P, \mu) M \diag(\mu\ast\alpha)^{-1}$ as described above; see \eqref{eq:codewords_equivalent_code}.
In order to show that we can use the cube code as distinguisher for such codes, we give upper bounds on the dimension of the space $S_k^{(3)}$ over $\mathbb F_q$, which is the linear space generated by three products in $S_k$.
Using the basis for $S_k$ given in \eqref{basis}, we consider the following products of three arbitrary basis elements:
\begin{align}\label{eq:cube}
    &(x_1^i + y x_2^i) ( x_1^j + y x_2^j) ( x_1^\ell + y x_2^\ell)\nonumber\\
    &= ( x_1^{i+j} +  y x_1^i x_2^j + y x_1^j x_2^i + y^2 x_2^{i+j}) ( x_1^\ell + y x_2^\ell)\nonumber\\
    &= x_1^{i+j+\ell}+  y x_1^{j+\ell} x_2^i +  y x_1^{i+\ell} x_2^j  +  y x_1^{i+j}x_2^\ell\nonumber\\
    & +  y^2 x_1^i x_2^{j+\ell}+  y^2 x_1^j x_2^{i+\ell}+  y^2 x_1^\ell x_2^{i+j}   + y^3 x_2^{i+j+\ell}
\end{align}
with $0 \leq i, j, \ell < k$.
Since these polynomials generate the space $S_k^{(3)}$, we will use the correspondence of these polynomials to triples $(i, j, \ell)$ later.

In the following we show an upper bound for the dimension of $S_k^{(3)}$ that is equal to the actual dimension we observed in all of our experiments, for $S_k^{(3)}$ as well as for the cube code of the public code; see Remark \ref{rem:experiments}.
Since the proof is a bit technical we will illustrate parts of the proof by means of an example and we encourage the reader to jump to Example \ref{blocks} when something is unclear.

\begin{theorem}\label{thm:dim_S_k}
    The dimension of $S_k^{(3)}$ over $\mathbb F_q$ is upper bounded by $2k^2-4k+4$ for $S_k = \{ f(x_1) + y f(x_2) \in \F_q[x_1, x_2,  y]: f \in \F_q[x]_{<k} \}$ and $k\geq 2$.
\end{theorem}

\begin{proof}
We identify the generators from equation \eqref{eq:cube} of $S_k^{(3)}$ with triples of the form $(i,j,\ell)$ where $0\leq \ell\leq j\leq i<k$ and order them (in non-decreasing order) according to their value of $i+j+\ell$.
Setting $s=i+j+\ell$ in \eqref{eq:cube}, we obtain that the set of monomials appearing in any of the elements of $S_k^{(3)}$ is a subset of
$$\bigcup_{s=0}^{3k-3}\{x_1^s, y x_1^s,y x_1^{s-1} x_2, y x_1^{s-2} x_2^2,\ldots,y x_2^s,y^2 x_1^s, y^2 x_1^{s-1} x_2,\ldots,y^2 x_2^s, y^3 x_2^s\}.$$

The dimension of the space $S_k^{(3)}$ is equal to the rank of the matrix $\calR$ whose columns correspond to these monomials and whose rows correspond to the generators from equation \eqref{eq:cube} of $S_k^{(3)}$, where the entry in row $(i,j,\ell)$ and column $y^{s_0} x_1^{s_1} x_2^{s_2}$ is equal to the coefficient of the monomial $y^{s_0} x_1^{s_1} x_2^{s_2}$ in the polynomial corresponding to $(i,j,\ell)$.

The matrix $\calR$ has only entries in $\{0,1,2,3\}$.
The rank of $\calR$ over $\F_q$ is upper bounded by the rank of $\calR$ considered as a matrix over $\mathbb{Z}$.
Therefore, we consider it as a matrix over $\mathbb{Z}$.

Ordering rows and columns according to $s$, we have that $\calR$ is a block diagonal matrix and hence, its rank is equal to the sum of the ranks of these $3k-2$ blocks.
Note that, for all $s\in\{0,\hdots,3k-3\}$, the column of block $s$ corresponding to $x_1^{s}$ is the all-one column.
The same is true for the column corresponding to $y^3x_2^{s}$, which we therefore can remove without changing the rank of $\calR$.
Also note that the column corresponding to $y^2 x_1^{s_1} x_2^{s_2}$ is identical to the column corresponding to $y x_1^{s_2} x_2^{s_1}$ 
and hence, the rank of $\calR$ does not change if we only consider the columns corresponding to 
\[
    \bigcup_{s=0}^{3k-3}\{x_1^s, y x_1^s, y x_1^{s-1} x_2, y x_1^{s-2} x_2^2, \ldots, y x_2^s\}.
\]
We first count the number of non-zero columns of each block and afterwards we will show that the rank for each block is upper bounded by its number of non-zero columns minus 2.

For $s\leq k-1$, the corresponding block has $s+2$ non-zero columns corresponding to the monomials $x_1^s, y x_1^s, y x_1^{s-1} x_2, \ldots, y x_2^s$.

For $k-1<s\leq 2(k-1)$, the corresponding block has $k+1$ non-zero columns corresponding to the monomials $x_1^s, y x_1^s, y x_1^{s-1} x_2, \ldots, y x_1^{s-(k-1)} x_2^{k-1}$ since we deleted the other identical non-zero columns.

For $2(k-1)<s\leq 3(k-1)$, the block has $3k-3-s+2$ non-zero columns corresponding to the monomials $x_1^s, y x_1^{2(k-1)} x_2^{s-2(k-1)}, \ldots, y x_1^{s-(k-1)} x_2^{k-1}$ since we deleted the other identical non-zero columns.

Note that the first column of each block is equal to the all-one column and the sum of the entries in each row of each block is equal to $4$ since we get $8$ terms when multiplying three brackets with two terms each in equation \eqref{eq:cube} and we already showed that we can delete half of the columns because each column appears twice.
Then, multiplying the first column with $3$ and subtracting all the other columns in each block gives the all-zero column.
This implies that the columns of each block are linearly dependent and any column can be deleted without changing the rank.

For $s=0$ and $s=3k-3$, the block consists only of one row (and two non-zero columns), corresponding to the polynomial where $(i, j, \ell)$ is equal to $(0,0,0)$ and $(k-1,k-1,k-1)$, respectively, and has rank one.

In the following, we assume $0<s<3k-3$, which also implies that we have at least $3$ non-zero columns per block, and show that in each of these blocks we can eliminate also a second column without changing the rank of the block.

Let $s=3z+d$ with $d\in\{0,1,2\}$.
The polynomials of the block correspond to $(i, j, \ell) = (z+r_1,z+r_2,z+r_3)$ with $k-z>r_1\geq r_2\geq r_3\geq-z\in\mathbb Z$, such that $r_1+r_2+r_3=d$.
We choose to eliminate the column corresponding to the monomial $y x_1^{s-z} x_2^{z}$ at first (as described above).

We remark here that from equation \eqref{eq:cube} we get that the coefficient of $y x_1^{i+j} x_2^\ell$ is equal to the maximum number of $i, j, \ell$ that coincide, i.e., it is $1$ if $\ell < j < i$, it is $2$ if $\ell = j, j < i$ or $i = j, \ell < j$ and it is $3$ if $i = j = \ell$.
Equivalently, this can be expressed in terms of $r_1, r_2, r_3$, where we have a $1$ if $r_3 < r_2 < r_1$, a $2$ if $r_3 = r_2, r_2 < r_1$ or $r_2 = r_1, r_3 < r_2$ and a $3$ if $r_3 = r_2 = r_1$.

Now, we multiply the column corresponding to $y x_1^{s-z-r} x_2^{z+r}$ with $r$ and sum up the results.
We claim that for the row $(i, j, \ell) = (z+r_1, z+r_2, z+r_3)$ this sum is equal to $r_1 + r_2 + r_3$.
This can be seen as follows.
If $r_3 < r_2 < r_1$, then the only non-zero entries in this row are in the columns corresponding to the monomials $x_1^s, y x_1^{s-z-r_1} x_2^{z+r_1}, y x_1^{s-z-r_2} x_2^{z+r_2}, y x_1^{s-z-r_3} x_2^{z+r_3}$ and these non-zero entries are all equal to $1$.
Multiplying these entries equal to $1$ with the values described above, i.e., with $r_1,r_2,r_3$, respectively, gives $r_1 + r_2 + r_3$ (note that we do not use the column for $x_1^s$).
The other cases for $r_1,r_2,r_3$ are dealt with similarly.
In this way we obtain a column all of whose entries are equal to $r_1+r_2+r_3=d$.
Therefore, the result is $d$-times the first column (i.e., the column corresponding to $x_1^s$).
Hence, due to this dependency, we can eliminate a second column in each block.

Therefore, the overall rank of $S_k^{(3)}$ is upper bounded by
\begin{align*}
&1+\sum_{s=1}^{k-1}s+(k-1)^2+\sum_{s=2k-1}^{3k-4}(3k-3-s)+1=2+(k-1)k+2\sum_{s=1}^{k-2}s\\
&=2k^2-4k+4,
\end{align*}
as we claimed.
\end{proof}

\begin{example}\label{blocks}
Let $k=6$, i.e., $s\in\{0,\hdots,15\}$.
We consider first what happens for $s=2$. The polynomials to be considered here are the ones corresponding to $(1,1,0)$ and $(2,0,0)$.
The monomials appearing in these polynomials, i.e., the monomials corresponding to non-zero columns in the block for $s=2$, are $x_1^2,y x_1^2, y x_1 x_2, y x_2^2, y^2 x_1^2, y^2 x_1 x_2, y^2 x_2^2, y^3 x_2^2$
and the corresponding block/submatrix inside the matrix $\calR$ is
      $$\begin{array}{cccccccccccccc}
     & x_1^2 &y x_1^2 & y x_1 x_2 & y x_2^2 & y^2 x_1^2 & y^2 x_1 x_2 & y^2 x_2^2 & y^3 x_2^2\\
    (1,1,0) & 1   & 1 & 2 & 0 & 0 & 2 & 1 & 1  \\
    (2,0,0)  & 1  & 2 & 0 & 1 & 1 & 0 & 2 & 1 & 
       \end{array}
      $$
 After eliminating the second half of columns (since they are all identical to one column from the first half) and observing that $z=0$ in this case, we obtain   
 $$\begin{array}{cccccccccccccc}
     & x_1^2 & y x_1^2 x_2^z & y x_1 x_2^{z+1} & y x_2^{z+2}\\
    (1,1,0) & 1   & 1 & 2 & 0 &  \\
    (2,0,0)  & 1  & 2 & 0 & 1 
       \end{array}
      $$
Multiplying the first column by $3$ and subtracting the other columns lets us eliminate the second column.
Moreover, we observe that the third column times $1$ plus the fourth column times $2$ gives two times the first column.
Therefore, the rank of this block is two, which in this very simple case of course could have been seen much easier considering rows instead of columns, but we use it to illustrate the general case, in which column operations are easier since the same kind of column operations works for each block.

Consider now in addition the case $s=7$.
    The corresponding block is 
    \small
    $$\begin{array}{cccccccccccccc}
     & x_1^7 & y x_1^7 & y x_1^6 x_2 & y x_1^5 x_2^2 & y x_1^4 x_2^3 & y x_1^3 x_2^4 & y x_1^2 x_2^5 & y x_1 x_2^6 & y x_2^7\\
    (3,2,2) & 1   & 0 & 0 & 2 & 1 & 0 & 0 & 0 & 0 \\
    (3,3,1)  & 1  & 0 & 1 & 0 & 2 & 0 & 0 & 0 & 0 & \\
    (4,2,1)& 1& 0 & 1 & 1 & 0 & 1 & 0 & 0 & 0 & \\
    (4,3,0) & 1 & 1 & 0 & 0 & 1 & 1 & 0 & 0 & 0 & \\
    (5,1,1) & 1 & 0 & 2 & 0 & 0 & 0 & 1 & 0 & 0\\
    (5,2,0) & 1 & 1 & 0 & 1 & 0 & 0 & 1 & 0 & 0
   \\
      & x_1^7 & y x_1^7 x_2^{z-2} & y x_1^6 x_2^{z-1} & y x_1^5 x_2^z & y x_1^4 x_2^{z+1} &  y x_1^3 x_2^{z+2} & y x_1^2 x_2^{z+3}& y x_1 x_2^{z+4} & y x_2^{z+5}
       \end{array}
      $$
      \normalsize
where we already deleted the second half of (identical) columns. Since $s>k-1$ in this case, the last two columns are zero and can be removed.
Eliminating afterwards two more columns with the column operations as described above with $d = r_1 + r_2 + r_3 = 1$, gives that the rank of this block is equal to $5$. Hence, for this block the rank is smaller than its number of rows. 
The existence of blocks with this property is the reason why the dimension of the cube code here is smaller than what is conjectured for the random case.
\end{example}

\begin{corollary}
\label{cor:upper_bound_public_code}
    The dimension of the cube code of the public code $\GRS_k(P, \mu) M$ is upper bounded by $\min\{n, 2 k^2 - 4k + 4\}$.
\end{corollary}

\begin{remark}\label{rem:experiments}
We verified with the computer that for $2 \leq k < 60$ the dimension of $S_k^{(3)}$ is actually equal to $2 k^2 - 4k + 4$ over the fields $
\Q, \F_2, \F_3, \F_4, \F_5, \F_7, \F_9, \F_{81}, \F_{97}, \F_{257}, \F_{625}, \F_{10007}$.
We did not find a field where this was not the case.
Moreover, in all our experiments for the attack and in additional experiments for $k < 60$, the actual dimension of the cube code of the public code was equal to $\min\{n, 2 k^2 - 4k + 4\}$.
\end{remark}

If instead of this best possible upper bound from Corollary \ref{cor:upper_bound_public_code} one is satisfied with a quadratic upper bound with worse constants, there are considerably shorter ways to prove such a bound without using Theorem \ref{thm:dim_S_k}.
For example one can use that the public code is a subcode of the sum $C_1 + C_2$ of two GRS codes $C_1, C_2$ of dimension $k$ each, as well as Lemma \ref{lem:power_codes} in the form
$$(C_1 + C_2)^{(3)} \subseteq C_1^{(3)} + C_1^{(2)} \ast C_2 + C_1 \ast C_2^{(2)} + C_2^{(3)},$$ 
where the dimension of the latter is upper bounded by $(3k - 2) + (2k - 1) k + k (2k - 1) + (3 k - 2)=4k^2+4k-4$.
We can also use this to get bounds on power codes of higher weight masked GRS codes.
\begin{lemma}
\label{lemma:higher_weight_masking_higher_power_codes}
    Let $M$ be an $n \times n$ matrix whose column weight is bounded from above by $w$.
    Then
    \[
        \dim ((\GRS_k(P, \mu) M)^{w+1}) = O(k^w).
    \]
\end{lemma}
\begin{proof}
    By definition of $M$ we have that $\GRS_k(P, \mu) M$ is contained in a sum of $w$ many $\GRS$ codes (here and only here we also allow column multipliers to be zero since this does not change the bounds on the dimension of powers of GRS codes).
    Let us denote them by $\GRS_k(Q_1, \mu_1), \ldots, \allowbreak \GRS_k(Q_w, \mu_w)$.
    Then, Lemma \ref{lem:power_codes} gives
    \begin{align*}
        &\left( \sum_{i=1}^w \GRS_k(Q_i, \mu_i) \right)^{(w+1)}\\
        &\subseteq \sum_{i_1, \ldots, i_{w+1} \in \{1, \ldots, w\}} \GRS_k(Q_{i_1}, \mu_{i_1}) \ast \ldots \ast \GRS_k(Q_{i_{w+1}}, \mu_{i_{w+1}}).
    \end{align*}
    Note that for each summand at least two of $i_1, \ldots, i_{w+1}$ are the same.
    Using in addition
    \[
        \dim \GRS_k(Q_{i_j}, \mu_{i_j})^{(2)} \leq 2 k - 1 = O(k),
    \]
    the result follows.
\end{proof}

\subsection{The Effect of Shortening}

In this section we discuss the effect of shortening and how to use it to extend the range of the distinguisher.

Recall from equation \eqref{eq:code_of_zeros_J} that for a code $C$ and a subset $J \subseteq \{1, \ldots, n\}$ we consider the code
\[
    C(J) = \{c \in C: c_i = 0 \text{ for all } i \in J\}.
\]
In addition recall that $\shortening_J(C) = \puncturing_J(C(J))$.
We assume that the codes never degenerate and that the code $C(J)$ has dimension $k - |J|$ where $k$ is the dimension of $C$.
Moreover, for ease of notation we assume using equations \eqref{eq:monomial_change} and \eqref{eq:mask_adding_i_and_pi_i} that $M$ is equal to $\hat{M}$ of the form
\[
    \hat{M}=\diag(\alpha)+\Pi \cdot \diag(\beta).
\]

First consider a single position at which the code is shortened.
For simplicity assume we shorten the code $\GRS_k(P, \mu) M$ in the first position.
Recall from equation \eqref{eq:codewords_equivalent_code} that we can go to an equivalent code with codewords $c$ such that $c_1 = f(P_1) + \nu_1 f(P_{\Pi(1)})$.

Note that if $f(P_1) = 0 = f(P_{\Pi(1)})$, then also $c_1 = 0$.
Therefore, the $(k-1)$-dimensional code $(\GRS_k(P, \mu) M) (\{1\})$ contains the code $\GRS_k(P, \mu)(\{1, \Pi(1)\}) M$ as a $(k-2)$-dimensional subcode.

More generally, let $J$ be the set of positions in which we shorten the code $\GRS_k(P, \mu) M$, i.e., we consider the code $(\GRS_k(P, \mu) M) (J)$.
Then, we have
\begin{equation}
\label{eq:subcode_shortening}
    \GRS_k(P, \mu)(\{J \cup \Pi(J)\}) M \subseteq (\GRS_k(P, \mu) M) (J).
\end{equation}
Note that the right hand side has dimension $k - |J|$ and the left hand side has dimension $k - |J \cup \Pi(J)|$, for which we know
\[
    |J| \leq |J \cup \Pi(J)| \leq 2 |J|.
\]
For illustration we start with a special case.
If $|J| = |J \cup \Pi(J)|$, then the two codes in equation \eqref{eq:subcode_shortening} are the same, i.e.,  
$$\GRS_k(P, \mu)(J ) M = (\GRS_k(P, \mu) M) (J).$$
\begin{lemma}
\label{lemma:J_eq_PJ_shorten_is_wt_2_masked}
If $J = \Pi(J)$, then 
$$\shortening_J(\GRS_k(P, \mu) M)=\shortening_J(\GRS_k(P,\mu))M_J$$
where $M_J$ is a matrix with row and column weight two that is obtained from $M$ by deleting all rows and columns with index in $J$.
\end{lemma}
\begin{proof}
Without loss of generality and for ease of notation assume that $J=\{1,\hdots,|J|\}$. 
Since $J=\Pi(J)$, the weight two mask $M$ is of the form $\begin{pmatrix}
    \ast & 0_{|J|\times (n-|J|)}\\
    0_{(n-|J|)\times |J|} & M_J
\end{pmatrix}$.
Denote by $G$ a generator matrix of $\GRS_k(P, \mu)(J)$ and by $G_J$ the generator matrix of 
$\shortening_J(\GRS_k(P,\mu))$ which satisfies
\[
    G = \begin{pmatrix}
0_{(k - |J|)\times |J|} & G_J 
\end{pmatrix}.
\]
We obtain,
\begin{align*}
G M&=\begin{pmatrix}
0_{(k - |J|) \times |J|} & G_J 
\end{pmatrix}\cdot \begin{pmatrix}
    \ast & 0_{|J|\times (n-|J|)}\\
    0_{(n-|J|)\times |J|} & M_J
\end{pmatrix}\\
&=\begin{pmatrix}
  0_{(k - |J|) \times|J|} &   G_JM_J
\end{pmatrix}.
\end{align*}
Since $\shortening_J(\GRS_k(P, \mu) M)$ is obtained by puncturing $GRS_k(P,\mu)(J) M$ in the positions in $J$, the result follows.
\end{proof}
Recall the following well-known lemma.
\begin{lemma}
\label{lemma:shorten_GRS_code}
    The code
    \[
        \shortening_J(\GRS_k(P,\mu))
    \]
    is a GRS code of length $n-|J|$ and dimension $k-|J|$.
\end{lemma}
Thus, in the setting of Lemma \ref{lemma:J_eq_PJ_shorten_is_wt_2_masked}, the shortening in positions $J$ gives another weight $2$ masked code.

In general, when $|J \cup \Pi(J)| = |J| + c$ with $0 \leq c \leq |J|$ we have that the codimension in \eqref{eq:subcode_shortening} is $c$.
This gives at least an upper bound for the dimension of the cube code of $(\GRS_k(P, \mu) M) (J)$.
Namely, let
\begin{equation}
\label{eq:shortening_subcode_complement}
    (\GRS_k(P, \mu) M) (J) = \GRS_k(P, \mu)(\{J \cup \Pi(J)\}) M + \calE
\end{equation}
where $\calE$ is a $c$-dimensional complement.
Then, we have
\begin{align}
\label{eq:cube_of_shortened_code}
    &((\GRS_k(P, \mu) M) (J))^{(3)} \subseteq\\
    &(\GRS_k(P, \mu)(\{J \cup \Pi(J)\}) M)^{(3)} + (\GRS_k(P, \mu)(\{J \cup \Pi(J)\}) M)^{(2)} \ast \calE \nonumber\\
    &+ (\GRS_k(P, \mu)(\{J \cup \Pi(J)\}) M) \ast \calE^{(2)} + \calE^{(3)}. \nonumber
\end{align}
\begin{remark}
\label{rmk:sum_not_direct}
    In our experiments we observed that the first two summands have non-trivial intersection.
    This means that also the bounds we derive below are not tight.
    Moreover, the bound we derive for $(\GRS_k(P, \mu)(\{J \cup \Pi(J)\}) M)^{(3)}$ to get the overall bound is also not tight.
    In our experiments we observed that one can shorten in way more positions than the bound below suggests.
    
    In the experiments it even seems as if asymptotically the quotient between the number of positions that we can shorten in and $n$ goes to $1$.
    However this does not mean that the distinguisher works for all rates (not even for constant rates) because the growth in the number of positions in which we can shorten is too slow for that.
\end{remark}

\begin{lemma}
    We have for the cube code of $(\GRS_k(P, \mu) M) (J)$ the following bound:
    \begin{align*}
        &\dim ((\GRS_k(P, \mu) M) (J))^{(3)} \\ &\leq
        2 k^2 - 4 k + 4 - (4 k (|J| + c)- 2 (|J| + c)^2 - 4 (|J| + c)) + \binom{k+1}{2} c + k \binom{c+1}{2} + \binom{c+2}{3}.
    \end{align*}
\end{lemma}
\begin{proof}
    By Lemma \ref{lemma:shorten_GRS_code} we know that $\GRS_k(P, \mu)(\{J \cup \Pi(J)\})$ is a GRS code of dimension $k - |J| - c$ and length $n - |J| - c$ with zero columns appended.
    In particular, the code  $\GRS_k(P, \mu)(\{J \cup \Pi(J)\}) M$ is equivalent to an evaluation code of $S_{k-|J|-c}^{(3)}$ with zero columns appended and thus, we can bound the dimension using Theorem \ref{thm:dim_S_k} as follows:
\begin{align*}
    &\dim (\GRS_k(P, \mu)(\{J \cup \Pi(J)\}) M)^{(3)} \leq 2 (k - |J| - c)^2 - 4 (k - |J| - c) + 4\\
    &= 2 (k^2 + (|J| + c)^2 - 2 k (|J| + c)) - 4k + 4(|J| + c) + 4\\
    &= 2 k^2 - 4 k + 4 - (4 k (|J| + c) - 2 (|J| + c)^2 - 4 (|J| + c)).
\end{align*}
Putting everything together we get
\begin{align*}
    &\dim ((\GRS_k(P, \mu) M) (J))^{(3)}\\ &\leq
    2 k^2 - 4 k + 4 - (4 k (|J| + c)
    - 2 (|J| + c)^2 - 4 (|J| + c)) + \binom{k+1}{2} \cdot c + k \cdot \binom{c+1}{2} + \binom{c+2}{3},
\end{align*}
as claimed.
\end{proof}

If this expression is smaller than $\min \left\{n - |J|, \binom{k - |J| +2}{3} \right\}$, we can distinguish using the cube code.
Note that in the special case $c = 0$, we get just the usual bound $2(k-|J|)^2-4(k-|J|)+4$ that we would also get from Lemma \ref{lemma:J_eq_PJ_shorten_is_wt_2_masked} together with Theorem \ref{thm:dim_S_k}.

\section{From Distinguishing to Key Recovery}
\label{sec:framework_distinguishing_key_recovery}

\subsection{Basic Idea}

Assume that we consider a certain class of $[n, k]_q$ codes that we can distinguish from random codes with the help of a certain distinguisher.
We assume that also monomial transformations of such codes are distinguishable.
Let $C$ be such a code and let $M \in \F_q^{n \times n}$ be an invertible matrix with average column or row weight $\delta \geq 1$.
Now, if we start with $\delta = 1$ and increase it, we expect the code $CM$ to get closer to random with respect to our distinguisher and our distinguisher to degrade.
Ideally, the probability of distinguishing does not change until it falls abruptly to $\frac{1}{2}$, but the information our distinguisher gives us comes closer to the random value.

Let us give an example.
Assume we can distinguish using the square code.
Ideally, the dimension of the square code increases when we increase $\delta$, but the probability that $CM$ is distinguishable does not substantially fall until we hit the point where the dimension of the square code of $CM$ is the dimension of the square code of a random code.
At this point it falls abruptly to $\frac{1}{2}$.

In our setting we actually start with a code $CM$ with $M$ having row weight $\delta > 1$.
If we have a corresponding distinguisher, we can try to use it to recover $M$ by considering $C M N$ for some very sparse $N$.
Depending on the form of $N$ we might or might not have cancellation in the product $M N$, i.e., the number of non-zero entries in $MN$ might or might not be smaller than or equal to the number of non-zero entries in $M$.
Then, we use our distinguisher to distinguish between the cases, i.e., different matrices $N$, with or without cancellation and in that way we get information about $M$.

In the next subsection, we give concrete choices for the matrix $N$ as above for the weight two masked McEliece cryptosystem with GRS codes.
Afterwards we use the information obtained from the distinguisher to recover a matrix that is in some sense equivalent to $M$, from which we can then recover the whole secret key.
The actual attack will only work for $n > 2 k^2 - 4k + 4$.

The basic idea of using cancellations when doing column operations for key recovery was already observed in \cite[Proposition 5]{couvreur2015polynomial} for $\delta<1 + R$ (where $R$ is the rate) using the square code.

However, there the authors use in addition puncturing to detect whether a column has weight $1$ or $2$.
After they have detected this, they turn columns of weight $2$ into columns of weight $1$ using cancellations in the columns.
Note that puncturing is a tool that does not seem to help in our situation since all our columns have weight $2$.

\subsection{Finding Cancellations}

As described in Section \ref{sec:Desc_McEliece_GRS_Weight_Two}, we are given a public code $\GRS_k(P, \mu) M$ where $M$ is an invertible matrix with row and column weight $2$.
We multiply on the right by a matrix of the form $I_n + \gamma E_{i, j}$ for $\gamma \in \F_q^\ast, 1\leq i \neq j\leq n$, where $E_{i,j}\in\mathbb F_q^{n\times n}$ is the matrix which has a one in position $i, j$ and zeros elsewhere.
We get the code 
\[
    \GRS_k(P, \mu) M (I_n + \gamma E_{i, j}).
\]
We consider now the matrix
\[
    M (I_n + \gamma E_{i, j}).
\]
This matrix has the same columns as $M$ except for the $j$-th column.
If we let $M_i$ denote the $i$-th column of $M$, then the $j$-th column of $ M (I_n + \gamma E_{i, j})$ is given by
\[
  M_j + \gamma M_i.
\]

Now there are several possibilities of what can happen:
\begin{itemize}
        \item If $M_i$ and $M_j$ have no common non-zero position, then $ M_j + \gamma M_i$ has weight $4$.
    \item 
If $M_i$ and $M_j$ have one common non-zero position, then generically $ M_j + \gamma M_i$ has weight $3$, but for a unique value of $\gamma$, it has weight $2$.
\item 
If both non-zero positions in $M_i$ and $M_j$ coincide, then $ M_j + \gamma M_i$ has generically weight $2$ and it has weight $1$ for two values of $\gamma$.
\end{itemize}
If we can somehow distinguish between several of these cases using the code $\GRS_k(P, \mu) M (I_n + \gamma E_{i, j})$, we can recover certain information about the matrix $M$.

We saw that the $j$-th column of $M (I_n + \gamma E_{i, j})$ can have weight $1, 2, 3, 4$.
The ideal situation would be if our distinguisher would allow us to distinguish between all these $4$ cases.

\begin{remark}
\label{rmk:experiments_wt_less_or_greater_distinguishing}
    In our experiments using the cube code as a distinguisher this was not the case.
    However, in the regime $n > 2 k^2 - 4k + 4$, it allowed us to distinguish between the two cases of at most $2$ and weight greater than $2$.
    The experiments gave us
    \begin{align}
\label{eq:cube_after_column_ops}
        &\dim (\GRS_k(P, \mu) M (I_n + \gamma E_{i, j}))^{(3)} \\
        \nonumber &= \begin{cases}
            2 k^2 - 4k + 4 & \text{if } M_j + \gamma M_i \text{ has weight at most $2$},\\
            2 k^2 - 4k + 5 & \text{if } M_j + \gamma M_i \text{ has weight greater than $2$}.
        \end{cases}
    \end{align}
    We will make the assumption that these are the cases we can distinguish from now on.
\end{remark}

Therefore, we can see three different cases with our distinguisher by going through all values of $\gamma \in \F_q^*$:

\begin{enumerate}
    \item There is no row in which the entries of the columns $M_i$ and $M_j$ are both non-zero.
    \item There is exactly one row where the columns $M_i$ and $M_j$ both have a non-zero entry and we know the $\gamma \in \F_q^*$ such that $M_j + \gamma M_i$ is zero in that row.
    \item The columns $M_i$ and $M_j$ have the non-zero entries in the same two rows.
\end{enumerate}

We detect Case 1. when the distinguisher gives us that for all values of $\gamma \in \F_q^*$ the weight of column $M_j + \gamma M_i$ is at least $3$.
We detect Case 2. if the distinguisher tells us that for $q-2$ values of $\gamma$ the column $M_j + \gamma M_i$ has weight at least $3$ and for a unique value of $\gamma$ the column $M_j + \gamma M_i$ has weight at most $2$.
This value is the $\gamma$ mentioned in Case $2$.
Finally, we see that we are in Case 3. if for all values of $\gamma$ the column $M_j + \gamma M_i$ has weight at most $2$.

\subsection{From Cancellation to Key Recovery}

In this section, we will show how to use the previous distinction to recover a weight two mask $\hat{M}$ and a GRS code $\hat{C}$ such that $\hat{C}\hat{M}=GRS_k(P,\mu)M$.

\begin{lemma}
\label{lemma:cycle_decomposition_matrix}
Let $M$ be a weight two matrix, i.e., a matrix where all rows and columns are of weight $2$, and let $\mathcal{G}$ be the bipartite graph with adjacency matrix $M$. Then, $\mathcal{G}$ can be decomposed into cycles.
\end{lemma}
\begin{proof}
Choose an arbitrary vertex $v_1$ in $\mathcal{G}$ and one of the two edges incident with $v_1$ in $\mathcal{G}$, denoted by $e$, and walk on a maximal path in $\mathcal{G}$ starting from $v_1$ along $e$ such that all the vertices in the path are distinct.
Since all vertices in $\mathcal{G}$ have degree $2$ and the number of vertices is finite, such a path is unique.
We denote this path by $[v_1, \ldots, v_\ell]$.
Let $v_\ell$ be the last vertex on this path.
Since the path cannot be extended, we know that the vertices adjacent to $v_\ell$ are $v_{\ell-1}$ and $v_i$ for some $i \in [1, \ell-2]$.
But if $i > 1$, then $v_i$ would be adjacent to $v_\ell, v_{i-1}, v_{i+1}$, a contradiction.
Thus, $v_\ell$ is adjacent to $v_1$.
Therefore, each vertex of $\mathcal{G}$ is contained in a unique cycle and $\mathcal{G}$ can be decomposed into cycles.
\end{proof}

Note that the cycle decomposition of the graph $\mathcal{G}$ as in the previous lemma
corresponds to the cycle decomposition of the permutation $\Pi$ in \eqref{P} and the submatrix of $M$ corresponding to a cycle of length $2\ell$ is up to row and column permutations equal to
\begin{equation}\label{eq:matrix_for_cycle}
   \begin{pmatrix}
a_1 & a_2 & 0 & \cdots &\cdots & 0 \\
0 & a_3 & a_4 & 0  &\cdots & \vdots\\
\vdots & \ddots & \ddots & \ddots  & \ddots  & \vdots \\
\vdots & & \ddots & \ddots & \ddots  & 0\\
0 & \cdots & \cdots & 0 & a_{2\ell-3} & a_{2\ell-2} \\
a_{2\ell} & 0 & \cdots & \cdots &0 & a_{2\ell-1}
\end{pmatrix}\in\mathbb F_q^{\ell\times \ell}
\end{equation}
with $a_i\in\mathbb F_q^\ast$ for $i \in \{1, \ldots, 2\ell\}$.

We will use the distinction into Cases $1.,2.,3.$ from the previous subsection to recover the cycle structure of the weight two mask.
Note that with Case $2.$ and Case $3.$ of the previous subsection we can actually find all the cycles.
In Case $2.$ we even know that if we have an entry $a$ in column $M_j$, then in column $M_i$ in the same row there is the entry $- \frac{a}{\gamma}$.
Since in Case $3.$, which corresponds to a $4$-cycle in the corresponding graph, we cannot distinguish whether or not cancellations occur, i.e., we cannot distinguish whether $M_j + \gamma M_i$ has weight $1$ or $2$, we cannot recover the values of the non-zero entries corresponding to $4$-cycles in this way.
These values will be recovered in a later step, which we will explain in the next subsection.

First, we apply again equation \eqref{eq:monomial_change}, where we now choose the monomial matrix $A$ such that $\hat{M}$ is a weight two mask that is ordered according to the cycles and is such that for each cycle the first non-zero entry in the first row of the cycle is a $1$ and when moving along the cycle each second non-zero entry is a $1$.

We illustrate this with an example.
\begin{example}
Consider, the weight two mask
$$\begin{pmatrix}
    a & 0 & 0 & 0 & 0 & b\\
     0 & 0 & 0 & c & d & 0\\
        e & 0 & f & 0 & 0 & 0\\
    0 & 0 & g & 0 & 0 & h\\
    0 & i & 0 & j & 0 & 0\\
       0 & k & 0 & 0 & l & 0
\end{pmatrix}.$$
Then, $\hat{M}$ is of the form
$$\begin{pmatrix}
    1 & 0 & 0 & 0 & 0 & \frac{b}{a}\\
    0 & 0 & \frac{g}{h} & 0 & 0 & 1\\
    \frac{e}{f} & 0 & 1 & 0 & 0 & 0\\
    0 & 1 & 0 & \frac{j}{i} & 0 & 0\\
    0 & 0 & 0 & 1 & \frac{d}{c} & 0\\
    0 & \frac{k}{l} & 0 & 0 & 1 & 0
\end{pmatrix}$$
where the first and last three rows each correspond to a $6$-cycle.
The values $\gamma$ together with the columns $i, j$ in Equation \eqref{eq:cube_after_column_ops} that lead to cancellation are in this example
\[
    \left(- \frac{b}{a}, 1, 6 \right), \left(- \frac{g}{h}, 6,3 \right), \left(- \frac{e}{f}, 3, 1 \right), \left(- \frac{j}{i}, 2, 4 \right), \left(-\frac{d}{c}, 4, 5 \right), \left(- \frac{k}{l}, 5, 2 \right).
\]
\end{example}

So we can assume that we know half of the non-zero entries of $\hat{M}$ already because we fixed them to be $1$.
But Case $2.$ of the previous section gives us the quotient between two non-zero entries in the same row.
In this case we also get the other value.

Therefore, we recovered all of the mask $\hat{M}$ except the non-zero values corresponding to $4$-cycles.
As mentioned above we will treat these values in the next subsection.
The previous steps are summarized in Algorithm \ref{alg:cancellations}.
In the case that no $4$-cycles appear the attack is essentially finished.

\begin{algorithm}[H]
\caption{\textproc{FindCancellations}$(C_{pub}, \mathcal{D}_{cube})$}
\label{alg:cancellations}
\begin{algorithmic}[1]
\Require Public $[n,k]_q$ code $C_{pub}$, cube code distinguisher $\mathcal{D}_{cube}$
\Ensure Set of disjoint cycles $W$, partial multiplier matrix $\Gamma \in \mathbb{F}_q^{n \times n}$
\State $W \gets \emptyset$, $\Gamma \gets 0_{n \times n}$, $\text{Visited} \gets \{ \text{False} \}_{i=0}^{n-1}$

\For{$start \gets 0 \text{ to } n-1$}
    \If{$\text{Visited}[start]$} \textbf{continue} \EndIf
    
    \State $w \gets (start)$, $\text{Visited}[start] \gets \text{True}$ 
    
    \While{\textbf{true}}
        \State $curr \gets \text{last}(w)$, $prev \gets \text{second-to-last}(w) \text{ (if } |w| \ge 2 \text{ else null)}$
        
        \For{$j \gets 0 \text{ to } n-1$}
            \If{$j == curr \textbf{ or } j == prev \textbf{ or } (\text{Visited}[j] \textbf{ and } j \neq start)$} 
                \State \textbf{continue} \Comment{Skip self, immediate backtrack, and visited columns}
            \EndIf
            
            \State $\Gamma_{valid} \gets \{ \gamma \in \mathbb{F}_q^* \mid \mathcal{D}_{cube}(C_{pub} \cdot (I_n + \gamma E_{curr,j})) = \text{True} \}$
            
            \If{$|\Gamma_{valid}| == q - 1$} \Comment{All $\gamma$ cancel, implying a 4-cycle}
                \State $w \gets w \parallel (j)$, $\text{Visited}[j] \gets \text{True}$, $W \gets W \cup \{w\}$
                \State \textbf{break while} 
                
            \ElsIf{$|\Gamma_{valid}| == 1$} 
                \State $\Gamma_{curr, j} \gets \text{sole element of } \Gamma_{valid}$
                \If{$j == start$} 
                    \State $W \gets W \cup \{w\}$ 
                    \State \textbf{break while} 
                \Else
                    \State $w \gets w \parallel (j)$, $\text{Visited}[j] \gets \text{True}$
                    \State \textbf{break for} 
                \EndIf
            \EndIf
        \EndFor
    \EndWhile
\EndFor
\State \Return $(W, \Gamma)$
\end{algorithmic}
\end{algorithm}

\subsection{Treating Four Cycles in the Weight Two Mask}

In this subsection we detail how to treat $4$-cycles.
First we give a very simple attack using known techniques.
Then we detail a heuristic approach which is more in spirit with the rest of the paper.
Finally, we give a partial treatment of the probability of $4$-cycles.

\subsubsection{Basic Idea and Simple Attack on the $4$-cycles}

We saw in Lemma \ref{lemma:cycle_decomposition_matrix} that a weight two mask can be decomposed into cycles and we already assumed that the rows of $\hat{M}$ are ordered according to these cycles using \eqref{eq:monomial_change}.
To also order the columns of $\hat{M}$ according to its cycle decomposition with a block diagonal structure, we multiply the public code on the right with a permutation matrix $\hat{\Pi}$, such that we end up with an equivalent code of the form
\[
    \GRS_k(P, \mu) M \hat{\Pi}=\hat{C}\hat{M}\hat{\Pi}.
\]
This means that $\hat{M}\hat{\Pi}$ is in block diagonal form where each cycle forms a block.
Let
\[
    \hat{M} \hat{\Pi} =
    \begin{pmatrix}
        \tilde{M} & 0 & 0 & \cdots & 0\\
        0 & \tilde{M}_{1} & 0 & \cdots & 0\\
        0 & 0 & \tilde{M}_2 & \cdots & 0\\
        \vdots & \vdots & \vdots & \ddots & \vdots\\
        0 & 0 & 0 & \cdots & \tilde{M}_r
    \end{pmatrix}
\]
with $\tilde{M}$ consisting of all cycles of length larger than $4$ and the $4$-cycles $\tilde{M}_1, \ldots, \tilde{M}_r$.
Note that we have already found $\tilde{M}$.
Then, multiplying the code $GRS_k(P,\mu) M\hat{\Pi}$ with
\begin{equation}
\label{eq:inverse_partial_matrix}
    \begin{pmatrix}
        \tilde{M}^{-1} & 0 & 0 & \cdots & 0\\
        0 & I_2 & 0 & \cdots & 0\\
        0 & 0 & I_2 & \cdots & 0\\
        \vdots & \vdots & \vdots & \ddots & \vdots\\
        0 & 0 & 0 & \cdots & I_2
    \end{pmatrix}
\end{equation}
gives the code
\begin{equation}
\label{eq:multiply_with_block_diagonal}
    \hat{C} \cdot \begin{pmatrix}
        I_{n - 2 r} & 0 & 0 & \cdots & 0\\
        0 & \tilde{M}_{1} & 0 & \cdots & 0\\
        0 & 0 & \tilde{M}_2 & \cdots & 0\\
        \vdots & \vdots & \vdots & \ddots & \vdots\\
        0 & 0 & 0 & \cdots & \tilde{M}_r
    \end{pmatrix}
\end{equation}
with $\hat{C}=GRS_k(P,\mu)\cdot A^{-1}$ for some monomial matrix $A$.
The inverse of the matrix in \eqref{eq:inverse_partial_matrix} is constructed in Algorithm \ref{alg:start_mask}.

\begin{algorithm}[H]
\caption{$\textsc{PartialMask}(W_{>4}, W_4, \Gamma)$}
\label{alg:start_mask}
\begin{algorithmic}[1]
\Require Set of cycles $W_{>4}$ ($|w| > 2$), Set of 4-cycles $W_4$ ($|w| == 2$), Multiplier matrix $\Gamma \in \mathbb{F}_q^{n \times n}$
\Ensure Partial block-diagonal mask $\hat{M}_{partial} \in \mathbb{F}_q^{n \times n}$ (the inverse of the matrix in \eqref{eq:inverse_partial_matrix}), Permutation matrix $\hat{\Pi} \in \mathbb{F}_q^{n \times n}$
\State $\Delta \gets ()$ 
\For{$w \in W_{>4}$}
    \State $\Delta \gets \Delta \parallel w$ 
\EndFor
\For{$w \in W_4$}
    \State $\Delta \gets \Delta \parallel w$ 
\EndFor

\State $\hat{\Pi} \gets 0_{n \times n}$
\For{$i \gets 0 \text{ to } n-1$}
    \State $\hat{\Pi}_{i, \Delta[i]} \gets 1$ 
\EndFor

\State $\hat{M}_{partial} \gets I_n$
\State $idx \gets 0$
\For{$w \in W_{>4}$}
    \State $L \gets |w|$
    \For{$i \gets 0 \text{ to } L-1$}
        \State $curr \gets idx + i$
        \State $next \gets idx + ((i + 1) \bmod L)$
        \State $\hat{M}_{partial}[curr, next] \gets -\Gamma_{w[i], w[(i+1) \bmod L]}$
    \EndFor
    \State $idx \gets idx + L$
\EndFor
\State \Return $(\hat{M}_{partial}, \hat{\Pi})$
\end{algorithmic}
\end{algorithm}

At this point we have several ways to finish the attack (assuming not everything is a 4-cycle and we have enough non 4-cycles).
One is to run the Sidelnikov-Shestakov or a filtration attack for all possible values of the $4$-cycles.
This gives a complexity of $(q-1)^2$ times the complexity of Sidelnikov-Shestakov for each $4$-cycle.
Another possibility would be to run the attack from \cite[Subsection "Case of remaining degree-2 positions"]{couvreur2015polynomial} to recover $\tilde{M}_1, \ldots, \tilde{M}_r$ (and the GRS code).
However, we did not implement this attack.
With $\hat{\Pi}$ we can then recover all of the mask.
In Section \ref{subsubsec:heuristic_square_attack_four_cycles} we will also detail a third method given in Algorithm \ref{alg:four_cycles}, which is more in the spirit of the rest of the paper.
But first, let us give an overview of the whole attack.

In summary, we go the following route, where arrows denote which code we are able to obtain from the preceding code:
\begin{align*}
    \begin{tikzcd}
        \GRS_k(P, \mu) M
        \arrow[d]\\
        \GRS_k(P, \mu) M \hat{\Pi}
        \arrow[d]\\
        \GRS_k(P, \mu) A^{-1}\begin{pmatrix}
        I_{n - 2 r} & 0 & 0 & \cdots & 0\\
        0 & \tilde{M}_{1} & 0 & \cdots & 0\\
        0 & 0 & \tilde{M}_2 & \cdots & 0\\
        \vdots & \vdots & \vdots & \ddots & \vdots\\
        0 & 0 & 0 & \cdots & \tilde{M}_r
    \end{pmatrix}
    \arrow[d]\\
    \GRS_k(P, \mu) A^{-1}.
    \end{tikzcd}
\end{align*}
The permutation $\hat{\Pi}$ we find using Algorithms \ref{alg:cancellations} and \ref{alg:start_mask}.
For the second step, we also use our attack and for the last any of the three possibilities of Sidelnikov-Shestakov for each of the $4$-cycles, the attack from \cite{couvreur2015polynomial} or the attack we detail next in Section \ref{subsubsec:heuristic_square_attack_four_cycles}.

\begin{remark}
\label{rmk:extend_attack_shortening}
    So far our attack is only applicable in the regime where $n > 2 k^2 - 4k + 4$.
    Using shortening, it seems reasonable to expect to be able to extend the attack to more parameters.
    This is, however, not straightforward.
    Remark \ref{rmk:experiments_wt_less_or_greater_distinguishing} no longer holds true in this setting.
    One would need a more intricate handling of the situation, for example, doing several column operations at once instead of only one.
\end{remark}

\subsubsection{A New Heuristic Way of Treating $4$-cycles}
\label{subsubsec:heuristic_square_attack_four_cycles}

This more heuristic, but simpler, attack is as follows.
We multiply each unknown $4$-cycle $(\begin{smallmatrix}
    1 & a\\
    b & 1
\end{smallmatrix})$ by a matrix of the form
$(\begin{smallmatrix}
    1 & \kappa\\
    \lambda & 1
\end{smallmatrix})$.
We are interested in the case where the product has exactly $2$ non-zero entries, and the second matrix $(\begin{smallmatrix}
    1 & \kappa\\
    \lambda & 1
\end{smallmatrix})$ is invertible.
The second matrix is invertible if and only if $\lambda \neq \frac{1}{\kappa}$.
Note that we have
\[
    \begin{pmatrix}
        1 & a\\
    b & 1
    \end{pmatrix}
    \begin{pmatrix}
        1 & \kappa\\
    \lambda & 1
    \end{pmatrix}
    = \begin{pmatrix}
        1 + a \lambda & \kappa + a\\
        b + \lambda & b \kappa + 1
    \end{pmatrix} =: \Lambda.
\]
Now, we want to study how zero entries in the resulting matrix $\Lambda$ can arise.
Note that the product of matrices also has to be invertible.
Assume that a diagonal entry of $\Lambda$ is zero, without loss of generality,
\[
    1 + a \lambda = 0
\]
or equivalently $\lambda = - \frac{1}{a}$.
In this case, the off-diagonal entries have to be non-zero in order for the product $\Lambda$ to be invertible.
Therefore, if we want exactly $2$ non-zero entries in $\Lambda$, we also need to have $b \kappa + 1 = 0$ and thus $\kappa = - \frac{1}{b}$.
Now, assume an off-diagonal entry of $\Lambda$ is zero.
By invertibility, again, the diagonal entries are non-zero.
To get exactly two non-zero entries in $\Lambda$ we have $\kappa = - a$ and $\lambda = - b$.

This implies that if we can detect which values for $\kappa$ and $\lambda$ lead to the case of exactly $2$ non-zero elements in $\Lambda$, we can recover $a, b$ up to order.

In Algorithm \ref{alg:four_cycles}, for each $i\in\{1,\hdots,r\}$, we multiply the code 
$$C' := \GRS_k(P, \mu) A^{-1}\begin{pmatrix}
        I_{n - 2 r} & 0 & 0 & \cdots & 0\\
        0 & \tilde{M}_{1} & 0 & \cdots & 0\\
        0 & 0 & \tilde{M}_2 & \cdots & 0\\
        \vdots & \vdots & \vdots & \ddots & \vdots\\
        0 & 0 & 0 & \cdots & \tilde{M}_r
    \end{pmatrix}$$
with a block diagonal matrix of the form
$$\begin{pmatrix}
        I_{n - 2(r-i+1) } & 0 & 0\\
         0 & \begin{array}{cc}
        1 & \kappa_i\\
    \lambda_i & 1
    \end{array} & 0 \\
   0 & 0 &  I_{2(r-i)}\\
    \end{pmatrix}$$
to obtain the code $C_i$ and apply the square-code as a distinguisher to $C_i$ in order to find the values for $\kappa_i$ and $\lambda_i$ such that $\tilde{M}_i\cdot \begin{pmatrix}
        1 & \kappa_i\\
    \lambda_i & 1
    \end{pmatrix}$ has exactly two non-zero entries.
\begin{remark}
    The heuristic we used for the square code  was the following:
    \begin{align*}
        &\dim (C_i)^{(2)}
    = \begin{cases}
        \dim (C')^{(2)} - 1 & \text{if there are exactly two non-zero entries in } \tilde{M}_i\cdot \begin{pmatrix}
        1 & \kappa_i\\
    \lambda_i & 1
    \end{pmatrix},\\
    \dim (C')^{(2)} & \text{otherwise.}
    \end{cases}
    \end{align*}
    This always worked in all our experiments from Section \ref{subsec:experiments}.
\end{remark}

\begin{remark}
    To decrease the complexity of the square code distinguisher, one could,  in addition, puncture $C_i$ in positions that get multiplied with an identity matrix in equation \eqref{eq:multiply_with_block_diagonal} as long as the length of the resulting code is still large enough to be square code distinguishable.
\end{remark}

\begin{algorithm}[H]
\caption{$\textsc{CompleteMask}(C_{base}, W_4, \hat{M}_{partial}, \mathcal{D}_{sq})$}
\label{alg:four_cycles}
\begin{algorithmic}[1]
\Require Base $[n,k]_q$ code $C_{base} := C_{pub} \cdot \hat{\Pi} \cdot \hat{M}_{partial}^{-1}$, $W_4$, Partial mask $\hat{M}_{partial} \in \mathbb{F}_q^{n \times n}$, Square distinguisher $\mathcal{D}_{sq}$
\Ensure Complete mask $\hat{M}_{\hat{\Pi}} := \hat{M} \hat{\Pi} \in \mathbb{F}_q^{n \times n}$
\State $\hat{M}_{\hat{\Pi}} \gets \hat{M}_{partial}$
\State $idx \gets n - 2 \cdot |W_4|$

\For{$w \in W_4$}
    \State $resolved \gets \text{False}$
    \For{$\kappa \in \mathbb{F}_q^*$}
        \For{$\lambda \in \mathbb{F}_q^* \setminus \{\kappa^{-1}\}$}
            \State $T \gets I_n + \kappa E_{idx, idx+1} + \lambda E_{idx+1, idx}$
            \If{$\mathcal{D}_{sq}(C_{base} \cdot T) == \text{True}$}
                \State $\hat{M}_{\hat{\Pi}}[idx, idx+1] \gets -\kappa$
                \State $\hat{M}_{\hat{\Pi}}[idx+1, idx] \gets -\lambda$
                \State $resolved \gets \text{True}$
                \State \textbf{break for}
            \EndIf
        \EndFor
        \If{$resolved$} \textbf{break for} \EndIf
    \EndFor
    \State $idx \gets idx + 2$
\EndFor
\State \Return $\hat{M}_{\hat{\Pi}}$
\end{algorithmic}
\end{algorithm}

Note that by reordering rows and multiplying with a non-zero scalar, it does not matter if we consider the mask with $4$-cycle
\[
    \begin{pmatrix}
        1 & a\\
        b & 1
    \end{pmatrix}
\]
or the mask with $4$-cycle
\[
    \begin{pmatrix}
        1 & \frac{1}{b}\\
        \frac{1}{a} & 1
    \end{pmatrix}.
\]
Therefore, both solutions that we can find for $\kappa$ and $\lambda$ work.

Using this, the overall algorithm is given in the following.

\begin{algorithm}[H]
\caption{$\textsc{KeyRecovery}(C_{pub}, \mathcal{D}_{cube}, \mathcal{D}_{sq})$}
\label{alg:full_attack}
\begin{algorithmic}[1]
\Require Public $[n,k]_q$ code $C_{pub}$, Distinguishers $\mathcal{D}_{cube}, \mathcal{D}_{sq}$
\Ensure Secret generator matrix $\hat{G}_{sec}$, Equivalent mask $\hat{M} \in \mathbb{F}_q^{n \times n}$

\State $(W, \Gamma) \gets \Call{FindCancellations}{C_{pub}, \mathcal{D}_{cube}}$
\State $W_{>4} \gets \{ w \in W \mid |w| > 2 \}$
\State $W_4 \gets \{ w \in W \mid |w| == 2 \}$

\State $(\hat{M}_{partial}, \hat{\Pi}) \gets \Call{PartialMask}{W_{>4}, W_4, \Gamma}$
\State $C_{base} \gets C_{pub} \cdot \hat{\Pi} \cdot \hat{M}_{partial}^{-1}$

\State $\hat{M}_{\hat{\Pi}} \gets \Call{CompleteMask}{C_{base}, W_4, \hat{M}_{partial}, \mathcal{D}_{sq}}$

\State $C_{GRS} \gets C_{pub} \cdot \hat{\Pi} \cdot \hat{M}_{\hat{\Pi}}^{-1}$
\State $\hat{G}_{sec} \gets \Call{SidelnikovShestakov}{C_{GRS}}$ 
\State $\hat{M} \gets \hat{M}_{\hat{\Pi}} \cdot \hat{\Pi}^{T}$ 

\State \Return $(\hat{G}_{sec}, \hat{M})$
\end{algorithmic}
\end{algorithm}

\subsubsection{Probabilities of $4$-cycles}

In our experiments we had some $4$-cycles, but not many.
A more thorough understanding of the probability of $4$-cycles would be beneficial.
A first step is given by the following lemma.

\begin{lemma}
 If the bipartite graph with adjacency matrix $M\in\mathbb F_q^{n\times n}$ can be decomposed into $N$ cycles of lengths $2\ell_1,\hdots 2\ell_N$ with $\ell_1+\hdots+\ell_N=n$, then the probability that $M$ is invertible is $\left(1-\frac{1}{q-1}\right)^N$.
\end{lemma}

\begin{proof}
Assume that $M$ is in block diagonal form where the blocks correspond to the cycles of its cycle decomposition according to Lemma \ref{lemma:cycle_decomposition_matrix} and Equation \eqref{eq:matrix_for_cycle}.
The determinant of the matrix in \eqref{eq:matrix_for_cycle} representing a $2\ell$-cycle is equal to 
\begin{align}\label{det}
\prod_{i=1,\ i\ \text{odd}}^{2\ell}a_i + (-1)^{\ell + 1} \prod_{i=1,\ i\ \text{even}}^{2\ell}a_i
\end{align}
with $a_i\neq 0$ for $i=1,\hdots,2\ell$. 
There are $(q-1)^{2\ell-1}$ possibilities for $a_1,\hdots,a_{2\ell}$ such that \eqref{det} is equal to zero, since we can choose all $a_i$ except one arbitrarily from $\mathbb F_q^\ast$ and then the last of the $a_i$ is fixed. Hence, the probability that a matrix corresponding to a $2\ell$-cycle has full rank is equal to $1-\frac{1}{q-1}$. Therefore, if the bipartite graph corresponding to $M$ can be decomposed into $N$ cycles, the probability that $M$ has full rank is equal to $\left(1-\frac{1}{q-1}\right)^N$.
\end{proof}

An application of Bayes' rule gives the probability of $4$-cycles.
However we have not found good bounds yet.

\subsection{Complexity}

All matrices in our algorithms have at most $n$ rows and columns.
Therefore, matrix multiplication and computing the rank can be done in $O(n^3)$.

The complexity of the attack in Algorithm \ref{alg:full_attack} can be computed as follows.
Algorithm \ref{alg:full_attack} employs the other algorithms and some matrix multiplications.
The main part of the complexity for reasonable values of $q$ comes from Algorithm \ref{alg:cancellations}.
Here, we have two for-loops of length $n$.
We also have a while-loop, but this only skips certain steps in the first for-loop.
We also need to use $q-1$ many times the cube code distinguisher.
For the cube code we need to compute the rank, what can be done in $O(n^3)$.
Therefore, everything in Algorithm \ref{alg:cancellations} can be done with $O(n^5 q)$ operations (technically over $\F_q$, so the bit complexity is slightly higher).
The complexity of Algorithm \ref{alg:start_mask} is $O(n)$.
Assuming that the number of $4$-cycles is constant, the complexity of Algorithm \ref{alg:four_cycles} is $q^2$ times the complexity of computing the rank, i.e., in total $O(q^2 n^3)$.
The matrix multiplications are in $O(n^3)$ and the Sidelnikov-Shestakov attack is in $O(k^3 + k^2 n)$, which in our case is $O(n^2)$ and thus, the overall complexity of Algorithm \ref{alg:full_attack} is $O(q n^5 + q^2 n^3)$.

This complexity is similar to the complexities of \cite{couvreur2014distinguisher, couvreur2015polynomial}.

It would be good to reduce the complexity in order to be able to run more experiments since in our setting $q n^5$ is already of the order of $k^{12}$.

\subsection{Experiments}
\label{subsec:experiments}

We ran the attack successfully on the following sets of parameters for $(n, k, q)$ (10 experiments each).
\begin{align*}
    (60, 6, 61), (60,6,64), (85, 7, 89), (90,7,121), (110, 8, 113), (112, 8, 125),\\
    (150, 9, 151), (180, 10, 181), (220, 11, 223), (255, 12, 257), (300, 13, 307).
\end{align*}

A toy implementation can be found here:
\url{https://git.math.uzh.ch/abmazu/an-attack-on-high-rate-mceliece-cryptosystems-using-generalized-reed-solomon-codes-with-weight-2-mask}.
For the larger parameter sets we parallelized the algorithm.

\section{Outlook and Open Problems}
\label{sec:outlook}

We divide this section into heuristics which we have used, but not proven for our distinguisher and attack on the one hand and open questions about more general attacks on the other hand.

\subsection{Unproven Assumptions}

We give a list of the assumptions that underlie our attacks but are not proven.
It might be worthwhile to try to prove all of them.
\begin{itemize}
    \item The dimension of the cube code of a random code of dimension $k$ and length $n$ is with high probability equal to $$\min\left\{\binom{k+2}{3},n\right\}.$$
    \item The dimension of $S_k^{(3)}$ over $\mathbb F_q$ is equal to $2k^2-4k+4$.
    \item The dimension of $(\GRS_k(P, \mu)M)^{(3)}$ is $2k^2-4k+4$.
    In particular the evaluation of $S_k^{(3)}$ at $(P_1, P_{\Pi(1)}, \nu_1), \ldots, (P_n, P_{\Pi(n)}, \nu_n)\in\mathbb F_q^3$ is injective.
    \item Equation \eqref{eq:cube_after_column_ops} holds.
    \item The probability of having many $4$-cycles is small.
    \item The heuristic Algorithm \ref{alg:four_cycles} always works.
\end{itemize}

\subsection{Other Distinguishers}

It is worth emphasizing again that we give a general framework that turns a good enough distinguisher into a key-recovery attack and for high rate use the cube code as distinguisher.
This immediately leads to the question if we can use other distinguishers here.

A first question is whether Betti numbers \cite{Randriambololona25} can help to extend the attack to smaller rates.
We guess that the answer will be a little, but not much.
The reason is as follows.
Let $C$ be a code of regularity $2$, i.e., with $C^{(2)} = \F_q^n$.
Then in \cite{Randriambololona25} it was shown that all Betti numbers except for $\beta_{1, 3}$ are determined by $I_2(C)$, the kernel of the evaluation map of quadratic forms.
So far, we have no evidence that there are special relations in $I_2(C)$ or that in this case the Betti numbers behave differently from random.
In this work we have dealt with the case of codes of regularity $\geq 4$, i.e., $C^{(3)} \neq \F_q^n$.
We think that the distinguisher and attack can be extended to codes of regularity $3$, i.e., codes with $C^{(2)} \neq \F_q^n$, but $C^{(3)} = \F_q^n$ using $\beta_{1, 4}, \beta_{2, 3}$.

Also other distinguishers such as the higher-order vanishing distinguisher \cite{hemmertwiemers25higherordervanishing} do not seem to be directly applicable.

\subsection{Open Problems}

Here we mention some open problems.
The first questions concern the range of applicability of our attack.
\begin{enumerate}
    \item The distinguisher can be improved using shortening. Is this also true for the attack? See Remark \ref{rmk:extend_attack_shortening}.
    \item Explain Remark \ref{rmk:sum_not_direct} and find the exact asymptotics of the number of positions in which we can shorten and the asymptotics of the distinguisher using shortening.
    \item Even with shortening the distinguisher applies only to a limited range of parameters.
    Are there other ways to attack in the lower rate regime?
    For example, is it possible to find subcodes $\GRS_{k'}(P, \mu) M$ where $k' < k$?
    In this case we could use the high rate attack on $\GRS_{k'}(P, \mu) M$.
    \item Related to the previous point: Can we use the cycle structure and Lemma \ref{lemma:J_eq_PJ_shorten_is_wt_2_masked} to extend the range of the attack?
    Assume we have shortened in a certain number of positions $J$ and for simplicity assume that the set $J \cup \Pi(J)$ lies within a single cycle.
    Further assume that after shortening in $J$ we are in the distinguishable regime.
    Then it seems highly likely that shortening in further positions we can detect using equations \eqref{eq:shortening_subcode_complement} and \eqref{eq:cube_of_shortened_code} whether the new positions are in $\Pi(J)$ or not.
    In this way one could obtain certain information about the cycle structure of the matrix $M$ even if one would be far from recovering all of $M$.
    \item Can we use other distinguishers as mentioned in the previous subsection?
    \item In this work we only attacked \cite{weger16} in the high rate regime.
    Is it possible to also attack schemes which have similar assumptions such as \cite{almeida2024convolutional, bolkema2017variations, Dyseryn26}, in the high rate regime?
    \item Can we also attack higher weight masked schemes with average weight $\delta > 2$?
    We know that in the high rate regime a sufficiently large power code will yield a distinguisher by Lemma \ref{lemma:higher_weight_masking_higher_power_codes} and if our heuristics about cancellation still hold, we can use it in the same way as before to recover information about the mask $M$.
    However, actual key recovery will be more complicated because the structure of $M$ is less straightforward.
    \item In the BBCRS scheme \cite{baldi2016enhanced} the authors used as a mask the sum of a weight two matrix and a matrix of rank $1$.
    Can we also break such a scheme in the high rate regime?
    \item In \cite{bolkema2017variations} the authors also propose to use AG codes with a weight $2$ mask.
    It seems very likely that the distinguisher and possibly also the attack will apply for general AG codes in the high rate regime.
\end{enumerate}

While some of the open problems seem relatively straightforward and some very hard we encourage the community to investigate all of them to have a more thorough understanding of the security of these problems.

\section*{Acknowledgements}

This work has been supported in part by the Swiss National Science Foundation under SNSF grant number 212865, and by the German research foundation, project number 513811367.

\bibliography{literature}
\bibliographystyle{plain}

\end{document}